\documentclass[11pt,english]{article}
\usepackage[utf8]{inputenc}
\usepackage{amsmath}
\usepackage{amsfonts}
\usepackage{amssymb}
\usepackage{amsthm}
\usepackage{thmtools}
\usepackage{xcolor}

\definecolor{ForestGreen}{rgb}{0.1333,0.5451,0.1333}
\definecolor{DarkRed}{rgb}{0.8,0,0}
\definecolor{Red}{rgb}{1,0,0}
\usepackage[linktocpage=true,
pagebackref=true,colorlinks,
linkcolor=DarkRed,citecolor=ForestGreen,
bookmarks,bookmarksopen,bookmarksnumbered]
{hyperref}
\renewcommand*\backref[1]{\ifx#1\relax \else (cit. on p. #1) \fi} 

\usepackage{cleveref}
\usepackage{geometry}
\usepackage{graphicx}

\newcommand{\N}{\mathbb{N}}
\newcommand{\Z}{\mathbb{Z}}
\newcommand{\I}{\mathbb{I}}
\newcommand{\F}{\mathbb{F}}

\newcommand{\adj}{\operatorname{adj}}

\newcommand{\cdeg}{\operatorname{cdeg}}
\newcommand{\diag}{\operatorname{diag}}
\newcommand{\poly}{\operatorname{poly}}

\newcommand{\matrixExponent}{{2.3729}} 

\declaretheorem[numberwithin=section,refname={Theorem,Theorems},Refname={Theorem,Theorems}]{theorem}

\declaretheorem[numberlike=theorem]{lemma}

\declaretheorem[numberlike=theorem]{corollary}
\declaretheorem[numberlike=theorem]{definition}

\begin{document}
	
\title{Sensitive Distance and Reachability Oracles for Large Batch Updates}

\newcommand{\email}[1]{#1}

\author{%
	\normalsize
	Jan van den Brand \thanks{KTH Royal Institute of Technology, \protect\email{janvdb@kth.se}.}
	\and 
	\normalsize
	Thatchaphol Saranurak \thanks{Toyota Technological Institute at Chicago, \protect\email{saranurak@ttic.edu}. Work mostly done while at KTH Royal Institute of Technology.}
}

\maketitle

\pagenumbering{gobble}

\begin{abstract}
	In the \emph{sensitive distance oracle} problem, there are three phases.
	We first preprocess a given \emph{directed} graph $G$ with $n$ nodes
	and integer weights from $[-W,W]$. Second, given a single batch of
	$f$ edge insertions and deletions, we update the data structure.
	Third, given a query pair of nodes $(u,v)$, return the distance from
	$u$ to $v$. In the easier problem called \emph{sensitive reachability
		oracle} problem, we only ask if there exists a directed path
	from $u$ to $v$.
	
	Our first result is a sensitive distance oracle with 
	$\tilde{O}(Wn^{\omega+(3-\omega)\mu})$
	preprocessing time, $\tilde{O}(Wn^{2-\mu}f^{2}+Wnf^{\omega})$ update
	time, and $\tilde{O}(Wn^{2-\mu}f+Wnf^{2})$ query time where the parameter
	$\mu\in[0,1]$ can be chosen. 
	The data-structure requires $O(Wn^{2+\mu} \log n)$ bits of memory.
	This is the first algorithm that can
	handle $f\ge\log n$ updates. Previous results (e.g. {[}Demetrescu
	et al. SICOMP'08; Bernstein and Karger SODA'08 and FOCS'09; Duan and
	Pettie SODA'09; Grandoni and Williams FOCS'12{]}) can handle at most
	2 updates. When $3\le f\le\log n$, the only non-trivial algorithm
	was by {[}Weimann and Yuster FOCS'10{]}. When $W=\tilde{O}(1)$, our
	algorithm simultaneously improves their preprocessing time, update
	time, and query time. In particular, when $f=\omega(1)$, their update
	and query time is $\Omega(n^{2-o(1)})$, while our update and query
	time are \emph{truly subquadratic }in $n$, i.e., ours is faster by
	a polynomial factor of $n$. To highlight the technique, ours is the
	first graph algorithm that exploits the \emph{kernel basis decomposition}
	of polynomial matrices by {[}Jeannerod and Villard J.Comp'05; Zhou,
	Labahn and Storjohann J.Comp'15{]} developed in the symbolic computation
	community. 
	
	As an easy observation from our technique, we obtain the first sensitive
	reachability oracle can handle $f\ge\log n$ updates. Our algorithm
	has $O(n^{\omega})$ preprocessing time, $O(f^{\omega})$ update time,
	and $O(f^{2})$ query time.
	This data-structure requires $O(n^2 \log n)$ bits of memory.
	Efficient sensitive reachability oracles were asked in {[}Chechik,
	Cohen, Fiat, and Kaplan SODA'17{]}.
	Our algorithm can handle any constant number of updates in constant time.
	Previous algorithms with constant
	update and query time can handle only at most $f\le2$ updates. Otherwise,
	there are non-trivial results for $f\le\log n$, though, with query
	time $\Omega(n)$ by adapting {[}Baswana, Choudhary and Roditty STOC'16{]}.
\end{abstract}

\newpage
\pagenumbering{arabic}

\section{Introduction}
\label{sec:introduction}

In the \emph{sensitive distance oracle} problem\footnote{In literature this setting is also referred to as \emph{distance sensitivity oracle} or \emph{emergency algorithm}.}, there are three phases. First, we preprocess a given \emph{directed}
graph $G$ with $n$ nodes, $m$ edges, and integer weights from $[-W,W]$.
Second, given a single batch of $f$ edge insertions and deletions,
we update the data structure. Third, given a pair of nodes $(u,v)$,
return the distance from $u$ to $v$. The time taken in each phase
is called \emph{preprocessing time}, \emph{update time}, and \emph{query
	time} respectively. In an easier problem called \emph{sensitive reachability
	oracle} problem,\emph{ }the setting is the same except that we only
ask if there exists a directed path from $u$ to $v$.

Although both problems are well-studied (as will be discussed below),
all existing non-trivial algorithms only handle $f\le\log n$ updates.
In contrast, in the analogous problems in undirected graphs, algorithms
that handle updates of any size are known. For example, there are
sensitive connectivity oracles \cite{PatrascuT07,DuanP10,DuanP17,HenzingerN16}
(i.e. reachability oracles in which is undirected graphs) with $\poly(n)$
preprocessing time and $\poly(f,\log n)$ update and query time
for \emph{every }$f$. For sensitive distance oracles in undirected
graphs (studied in \cite{BaswanaK13,ChechikLPR12,ChechikCFK17}),
Chechik, Langberg, Peleg, and Roditty \cite{ChechikLPR12} show that
this is possible as well if \emph{approximate} distance is allowed.
It is interesting whether there is any inherent barrier for $f\ge\log n$
in directed graphs. In this paper,
we show that there is no such barrier.

\paragraph{Sensitive distance oracle.}

All previous works on this problem handle edge \emph{deletions} only\footnote{Node deletions can be handled as well using a simple reduction.}.
The first result is in the case when $f=1$ by Demetrescu et al.~\cite{DemetrescuTCR08}\footnote{The result was published in 2008, but previously announced at SODA
	and ISAAC 2002.} since 2002. Their algorithm has $O(mn^{1.5})$ preprocessing time
and $O(1)$ update and query time. This preprocessing time is improved
by Bernstein and Karger to $\tilde{O}(n^{2}\sqrt{m})$ \cite{BernsteinK08}
and finally to $\tilde{O}(mn)$ \cite{BernsteinK09}. The $\tilde{O}(mn)$
bound is optimal up to a sub-polynomial factor unless there is a truly
subcubic algorithm for the all-pairs shortest path problem. Later,
Grandoni and Williams \cite{GrandoniW12} showed a new trade-off.
For example, they can obtain $O(Wn^{2.88})$ preprocessing time and
$O(n^{0.7})$ update and query time. 

It was stated as an open question in \cite{DemetrescuTCR08,BernsteinK08,BernsteinK09}
whether there is a non-trivial algorithm for handling more than one
deletion. Duan and Pettie \cite{DuanP09a} give an affirmative answer
for $f=2$. Their algorithm has $\mbox{poly}(n)$ preprocessing time
and $\tilde{O}(1)$ update and query time. They however stated that
``it is practically infeasible'' to extend their algorithm to handle
three deletions. Weimann and Yuster \cite{WeimannY13} later showed
an algorithm for handling up to $f\le\log n/\log\log n$ deletions. For any parameter $\mu\in[0,1]$, their algorithm has $\tilde{O}(Wn^{\omega+\mu})$
preprocessing time and $\tilde{O}(n^{2-\mu/f})$ update and query time.
This algorithm remains the only non-trivial algorithm when $f\ge3$. 

To summarize, no previous algorithm can handle $f > \log n/\log\log n$ updates.
Another drawback of all previous algorithms is that they inherently cannot handle changes of edge weights.\footnote{One restricted and inefficient approach for handling weight increases via edge deletions is as follows. For each edge $e$, we need to know all possible weights for $e$, say $w_1\le\dots\le w_k$. Then, we add multi-edges $e_1,\dots,e_k$ with the same endpoints as $e$, and the weight of $e_i$ is $w_i$. To increase the weight from $w_1$ to $w_j$, we delete \emph{all} edges $e_1,\dots,e_{j-1}$. However, all previous algorithms handle only small number deletions.}
In the setting where we allow changes of edge weights, it is quite natural to ask for algorithms that handle somewhat large batches of updates.\footnote{For example, in a road network, there usually are not too many completely blocked road sections (i.e. deleted edges), but during rush-hour there can be many roads with increased congestion (i.e. edges with increased weight).}

In this work, we show the first sensitive distance oracle that can handle \emph{any} number of updates. 
Moreover, we allow both edge insertions, deletions, and changes of weights. Ours is the fastest oracle for
handling $f\ge3$ updates when $W=\tilde{O}(1)$: 
\begin{theorem}
	\label{thm:sense st}For any parameter $\mu\in[0,1]$, there is a
	Monte Carlo data structure that works in three phases as follows: 
	\begin{enumerate}
		\item Preprocess a directed graph with $n$ nodes and integer weights from
		$[-W,W]$ in $\tilde{O}(Wn^{\omega+(3-\omega)\mu})$ time and store $O(Wn^{2+\mu} \log n)$ bits.
		\item Given updates to any set of $f$ edges, update the data structure
		in $\tilde{O}(Wn^{2-\mu}f^{2}+Wnf^{\omega})$ time and store additional $O(Wnf^2 \log n)$ bits. 
		\item Given a pair of nodes $(u,v)$, return the distance from $u$ to $v$
		or report that there is a negative cycle in $\tilde{O}(Wn^{2-\mu}f+Wnf^{2})$
		time. 
	\end{enumerate}
\end{theorem}

The algorithm is Monte Carlo randomized and is correct with high probability,
i.e. the algorithm may return a larger distance with probability at
most $1/n^{c}$ for any constant $c$. 
We note that, any algorithm that can handle edge deletions can also handle node deletions by using a standard reduction. 

Let us compare this result to the algorithm by Weimann and Yuster
\cite{WeimannY13}. First, instead of edge deletions only, we allow
\emph{both} edge insertions, deletions, and changes of weight. 
The second point is efficiency. For any $\mu\in[0,1]$, we improve the
preprocessing time from $\tilde{O}(Wn^{\omega+\mu})$ to $\tilde{O}(Wn^{\omega+(3-\omega)\mu})$.
Both of our update and query time are bounded above by $\tilde{O}(Wn^{2-\mu}f^{\omega})$.
When $W=\tilde{O}(1)$, this improves their bound of $\tilde{O}(n^{2-\mu/f})$
for any $f\ge2$ by a polynomial factor. In particular, when $\mu>0$
and $f=\omega(1)$, their bound is $\Omega(n^{2-o(1)})$, while ours
is\emph{ truly subquadratic }in $n$. The last and most conceptually
important point is that we remove the constraint that $f\le\log n/\log\log n$.
While their technique is inherently for small $f$, our approach is
different and can handle any number of updates.

To highlight our technique, this is the first graph algorithm that
exploits the powerful \emph{kernel basis decomposition} of polynomial
matrices by Jeannerod and Villard \cite{JeannerodV05} and Zhou, Labahn
and Storjohann \cite{ZhouLS15} from the symbolic computation community.
We explain the overview of our algebraic techniques in \Cref{sec:overview}.
We then compare our techniques to the previous related works in \Cref{sub:compare}.

\paragraph{Sensitive reachability oracle.}

It was asked as an open problem by Chechik et al. \cite{ChechikCFK17}
whether there is a much more efficient data structure if the query
is only about reachability between nodes and not distance. As an easy
observation from our technique, we give a strong affirmative answer
to this question:
\begin{theorem}
	\label{thm:reach} There is a Monte Carlo sensitive reachability oracle that preprocess
	an $n$-node graph in $O(n^{\omega})$ time and stores $O(n^2 \log n)$ bits. Then, given a set of
	$f$ edge insertions/deletions and node deletions, update the data
	structure in $O(f^{\omega})$ time and store additional $O(f^2 \log n)$ bits. 
	Then, given a query $(u,v)$, return if there is directed path from $u$ to $v$ in $O(f^{2})$
	time. 
\end{theorem}

Previously, there are only algorithms that handle $f=1$
edge deletions based on dominator trees \cite{LengauerT79,BuchsbaumGKRTW08,GeorgiadisT12,FraczakGMT14}
or $f\le2$ edge deletions \cite{Choudhary16}. These algorithms have
$O(1)$ update and query time. Another approach based on fault-tolerant
subgraphs\footnote{The goal is to find sparse subgraphs preserving reachability information
	of the original graph even after some edges are deleted.} gives algorithms with query time at least $\Omega(n)$ \cite{BaswanaCR15,BaswanaCR16}
and becomes trivial when $f\ge\log n$. When $f=O(1)$, our algorithm has $O(1)$
update and query time like the algorithms using the first
approach. Moreover, ours is the first which handles updates of any
size. 

It was shown in \cite{AbboudW14,HenzingerL0W-ITCS17} that, assuming
the \emph{Boolean Matrix Multiplication} conjecture, there cannot
be any constant $\epsilon>0$ and a ``combinatorial'' algorithm
for \Cref{thm:reach} which has $O(n^{3-\epsilon})$ preprocessing
time, and can handle $2$ edge insertions using $O(n^{2-\epsilon})$
update and query time. Our result does not refute the conjecture as
we use fast matrix multiplication.

\subsection{Technical overview}
\label{sec:overview}

\paragraph{Set the stage.}
The first step of all our results is to reduce the problem on graphs to algebraic problems using the following known reduction (see \Cref{lem:shortestDistanceReduction} for a more detailed statement).
Let $G$ be an $n$-node graph with integer weights from $[-W,W]$.
Let $\F$ be a finite field of size at least $n^3$.
We construct a polynomial matrix $A \in \F[X]^{n \times n}$ such that  
$A_{i,j} = a_{i,j} X^{W+c_{i,j}}$ where $c_{i,j}$ is the weight of edge $(i,j)$ and $a_{i,j}$ is a random element from $\F$.
Then, with high probability, we can read off the distance from $i$ to $j$ in $G$ from the $(i,j)$ entry of the adjoint matrix $\adj(A)$ of $A$, for all pairs $(i,j)$.
That is, it suffices to build a data structure on the above polynomial matrix $A$ that can handle updates and can return an entry of its adjoint $\adj(A)$. Updating $f$ edges in $G$ corresponds to adding $A$ with $C$ where $C\in \F[X]^{n \times n}$ has $f$ non-zero entries. Further we have with high probability that $\det(A) \neq 0$, so it is enough to focus on algorithms that work on non-singular matrices.
From now, we let $\F$ be an arbitrary field and we use the number of field operations as complexity measure.

\paragraph{Warm up: slow preprocessing.}
To illustrate the basic idea how to maintain the adjoint, we will prove the following:

\begin{lemma}
	\label{lem:simpleAdjoint}
	Let $A \in \F[X]^{n \times n}$ be a polynomial matrix of degree $d$ with $\det(A) \neq 0$, then there exists an algorithm that preproceses $A$ in $\tilde{O}(dn^{\omega+1})$ operations. Then, for any $C \in \F[X]^{n \times n}$ with $f$ non-zero entries of degree at most $d$, we can query any entry of $\adj(A+C)$ in $\tilde{O}(dnf^{\omega+1})$ operations, if $\det(A+C) \neq 0$.
	
	If $d=0$, then the preprocessing and query time are $O(n^\omega)$ and $O(f^\omega)$ respectively.
\end{lemma}

This immediately implies a weaker statement of \Cref{thm:sense st} when $\mu = 1$ and the edge updates and the pair to be queried are given at the same time. 
We remark that, from the simple proof below, this already gives us the first non-trivial sensitive distance oracle which can handle any number of $f$ updates. Previous techniques inherently require $f \le \log n / \log \log n$.
As the reachability problem can be considered a shortest path problem, where every edge has weight zero, we also obtain a result similar to \Cref{thm:reach} from \Cref{lem:simpleAdjoint} for $d=0$. 

\begin{corollary}\label{cor:simpleResult}
	Let $G$ be some directed graph with integer weights in $[-W,W]$, then there exists an algorithm that preprocess $G$ in $\tilde{O}(Wn^{\omega+1})$ time. Then, given $f$ edge updates to $G$ and a query pair $(u,v)$, it returns the distance from $u$ to $v$ in the updated graph in $\tilde{O}(Wnf^{\omega+1})$ time.
\end{corollary}

\begin{corollary}\label{cor:simpleResultReachability}
	Let $G$ be some directed graph, then there exists an algorithm that preprocess $G$ in $\tilde{O}(n^{\omega})$ time. Then, given $f$ edge updates to $G$ and a query pair $(u,v)$, it returns the reachability from $u$ to $v$ in the updated graph in $\tilde{O}(f^{\omega})$ time.
\end{corollary}

To prove \Cref{lem:simpleAdjoint}, we use the key equality below based on the Sherman-Morrison-Woodbury formula. The proof is deferred to \Cref{sec:proof adjointUpdate}. 

\begin{lemma}\label{lem:adjointUpdate}
	
	Let $A$ be an $n \times n$ matrix and $U,V$ be $n \times f$ matrices, such that $\det(A),\det(A+UV^\top ) \neq 0$. Define the $f \times f$ matrix $M:=\I \cdot \det(A) + V^\top \adj(A) U$. Then, we have 
	$$
	\adj(A+UV^\top ) = \frac{\adj(A)\det(M) - \left(\adj(A) U\right) \: \adj(M) \: \left(V^\top \adj(A)\right)} {\det(A)^f }.
	$$
	
\end{lemma}

The algorithm for \Cref{lem:simpleAdjoint} is as follows:
We preprocess $A$ by computing $\det(A)$ and $\adj(A)$. As $\det(A)$ and $\adj(A)$ have degree at most $dn$, this takes $\tilde{O}(dn \times n^{\omega}) = \tilde{O}(dn^{\omega+1})$ field operations \cite[Chapter 1]{BuergisserCS97} (or just $O(n^\omega)$ if $d=0$).
Next, we write $C = UV^\top$ where $U,V \in \F[X]^{n \times f}$ have only one non-zero entry of degree $\le d$ per column.
To compute $\adj(A + UV^\top)_{i,j}$, we simply compute
\begin{align}
\adj(A)_{i,j}\frac{\det(M)}{\det(A)^f}-
\frac{ 
\overbrace{\left(\vec{e}_i^\top \adj(A) U\right)}^{:= \vec{u}}
 \: \adj(
\overbrace{\I \cdot \det(A) + V^\top \adj(A) U}^{=:M} 
) \: 
\overbrace{\left(V^\top \adj(A) \vec{e}_j \right)}^{:= \vec{v}}
 } {\det(A)^f }\label{eqn:entryQueryExample}
\end{align}
where $\vec{e}_i$ is the $i$-th standard unit vector.
This computation can be separated into the following steps:
\begin{enumerate}
\item Compute $\adj(A)_{i,j}$, $\vec{u} := \vec{e}_i^\top \adj(A) U$ and $\vec{v} := V^\top \adj(A) \vec{e}_j$. Note that because of the sparsity of $U$ and $V$, $\vec{u}$ and $\vec{v}$ are essentially just vectors of $f$ elements of $\adj(A)$, each multiplied by a non-zero element of $U$ and $V$. Thus this step can be summarized as obtaining $O(f)$ elements of $\adj(A)$. \label{step:readfelements}
\item Compute $V^\top \adj(A) U$ which are likewise just $f^2$ elements of $\adj(A)$, each multiplied by a non-zero element of $U$ and $V$. So this time we have to obtain $O(f^2)$ entries of $\adj(A)$.\label{step:readf2elements}
\item Compute the adjoint $\adj(M)$ and determinant $\det(M)$. \label{step:computeAdjoint}
\item Compute $\det(A)^f$. \label{step:power}
\item Compute $\adj(A)_{i,j}\frac{\det(M)}{\det(A)^f}$ and vector-matrix-vector product $\vec{u}\adj(M)\vec{v}$ and divide it by $\det(A)^f$. Then subtract the two values and we obtain $\adj(A+UV^\top )_{i,j}$.\label{step:expensive}
\end{enumerate}

Steps \ref{step:readfelements} and \ref{step:readf2elements} require only $\tilde{O}(dnf^2)$ field operations as we just have to read $O(f^2)$ entries of $\adj(A)$ and multiply them by some small $d$-degree polynomials from $U$ and $V$. 
In step \ref{step:computeAdjoint} we have to compute the adjoint and determinant of a $f \times f$ matrix of degree $dn$. 
This takes $\tilde{O}(dnf^{\omega+1})$ operations.
Step \ref{step:power} computes $\det(A)^f$ where $\det(A)$ is of degree $dn$, which takes $\tilde{O}(dnf)$  operations. 
Step \ref{step:expensive} takes $\tilde{O}(dnf^3)$ because $\adj(M)$ is a degree $dnf$ matrix of dimension $f \times f$. The total number of operations is thus $\tilde{O}(dnf^{\omega+1})$. The algorithm does not require the upper bound of $f \le \log n / \log \log n$ as in \cite{WeimannY13}.

For the reachability case, when $d=0$, all entries of the matrices and vectors are just field elements, so steps \ref{step:readfelements} and \ref{step:readf2elements} need only $O(f^2)$ operations. Step \ref{step:computeAdjoint} needs $O(f^\omega)$ operations, while \ref{step:power} can be done in just $O(\log f)$ operations, and the last step \ref{step:expensive} requires only $O(f^2)$ operations.

\paragraph{Key technique: kernel basis decomposition.}
The biggest bottleneck in \Cref{lem:simpleAdjoint} is explicitly computing $\adj(A)$ and $\adj(M)$. For $\adj(A)$ it already takes $\Omega(dn^3)$ operations in the preprocessing step just to write down the $n^2$ entries of $\adj(A)$ each of which has degree upto $dn$.
What we need is an \emph{adjoint oracle}, i.e. a data structure on $A$ with fast preprocessing that can still quickly answer queries about entries of $\adj(A)$.

By replacing this data structure in the five steps of the proof of \Cref{lem:simpleAdjoint}, this immediately gives distance oracles in the sensitive setting which dominate previous results when $W = \tilde{O}(1)$. 
(This is how we obtain \Cref{thm:sense st}).

The key contribution of this paper is to realize that the technique in \cite{JeannerodV05,ZhouLS15} actually gives the desired adjoint oracle.
This technique, which we call the \emph{kernel basis decomposition}, is introduced by Jeannerod and Villard \cite{JeannerodV05} and then improved by Zhou, Labahn, and Storjohann \cite{ZhouLS15}. 
It is originally used for inverting a polynomial matrix of degree $d$ in $\tilde{O}(dn^3)$ operations. 
However, the following adjoint oracle is implicit in Section 5.3 of \cite{ZhouLS15}\footnote{Section 5.3 of \cite{ZhouLS15} discusses computing $v^\top A^{-1}$, where the result is given by a vector $u$ with entries of the form $p/q$ where $p,q$ are polynomials of degree at most $O(dn)$.
	Since $\det(A)$ can be computed in $\tilde{O}(dn^\omega)$ \cite{Storjohann03,LabahnNZ17} and $\adj(A) = A^{-1} \det(A)$, we get \Cref{thm:kernelBasePreprocessing}.
}:

\begin{theorem}
	\label{thm:kernelBasePreprocessing}
	There is a data-structure that preprocesses $B \in \F[X]^{n \times n}$ where $\det(B) \neq 0$ and $\deg(B) \le d$ in $\tilde{O}(dn^\omega)$ operations.
	Then, given any $\vec{v} \in \F[X]^n$ where $\deg(\vec{v}) \le d$, it can compute $\vec{v}^\top \adj(B)$ in $\tilde{O}(dn^2)$ operations.
\end{theorem}

However, to get $o(Wn^2)$ query time as in \Cref{thm:sense st}, \Cref{thm:kernelBasePreprocessing} is not enough. 
Fortunately, by modifying the technique from \cite{ZhouLS15} in a white-box manner (see \Cref{sec:inverseOracle} for details), we can obtain the following trade-off which is essential for \Cref{thm:sense st}. The result essentially interpolates the exponents of the following two extremes: $\tilde{O}(dn^3)$ preprocessing and $O(dn)$ query time when computing the adjoint explicitly, or $\tilde{O}(dn^\omega)$ preprocessing and $O(dn^2)$ query time when using \Cref{thm:kernelBasePreprocessing}.

\begin{theorem}\label{thm:extendedKernelBasePreprocessing}
	For any $0 \le \mu \le 1$, 
	there is a data-structure that preprocesses $B \in \F[X]^{n \times n}$ where $\det(B) \neq 0$ and $\deg(B) \le d$ in  
	$O(dn^{\omega  +(3-\omega)\mu})$ operations.
	Then, given any pair $(i,j)$, it returns $\adj(B)_{i,j}$ in $O(dn^{2-\mu})$ operations.
\end{theorem}

To see the main idea, we give a slightly oversimplified description of the oracle in \Cref{thm:kernelBasePreprocessing} which allows us to show how to modify the technique to obtain \Cref{thm:extendedKernelBasePreprocessing}.
Below, we write a number for each matrix entry to indicate a bound on the degree, e.g. when we write $(4,4,4)^\top$ then we mean a 3-dimensional vector with entries of degree at most 4.

Suppose that we are now working with an $n \times n$ matrix $B$ of degree $d$.  Then \cite{ZhouLS15} (and \cite{JeannerodV05} for a special type of matrices) is able to find a full-rank matrix $A$ of degree $d$ in $\tilde{O}(dn^\omega)$ field operations, such that
\begin{align*}
\underbrace{\left(
	\begin{array}{cccc}
	d&d&d&d\\
	d&d&d&d\\
	\hline
	d&d&d&d\\
	d&d&d&d
	\end{array}
	\right)}_{B}
\underbrace{\left(
	\begin{array}{cc|cc}
	d&d&d&d\\
	d&d&d&d\\
	d&d&d&d\\
	d&d&d&d
	\end{array}
	\right)}_{A}
=
\left(
\begin{array}{cc|cc}
2d&2d&  &  \\
2d&2d&  &  \\
\hline
&  &2d&2d\\
&  &2d&2d
\end{array}
\right).
\end{align*}
Here the empty sections of the matrices represent zeros in the matrix and the $d$s represent entries of degree at most $d$.
This means the left part of $A$ is a \emph{kernel-base} of the lower part of $B$ and likewise the right part of $A$ is a \emph{kernel-base} of the lower part of $B$.

This procedure can now be repeated on the two smaller $n/2 \times n/2$ matrices of degree $2d$. After $\log n$ such iterations we have:
\begin{align*}
\underbrace{\left(
	\begin{array}{cccc}
	d&d&d&d\\
	d&d&d&d\\
	\hline
	d&d&d&d\\
	d&d&d&d
	\end{array}
	\right)}_{B}
\underbrace{\left(
	\begin{array}{cc|cc}
	d&d&d&d\\
	d&d&d&d\\
	d&d&d&d\\
	d&d&d&d
	\end{array}
	\right)}_{A_1}
\underbrace{\left(
	\begin{array}{cc|cc}
	2d&2d&  &  \\
	2d&2d&  &  \\
	\hline
	&  &2d&2d\\
	&  &2d&2d
	\end{array}
	\right)}_{A_2}
\cdots
=
\underbrace{\left(
	\begin{array}{cccc}
	\cline{1-1}
	\multicolumn{1}{|c|}{dn}& & & \\
	\cline{1-2}
	&\multicolumn{1}{|c|}{dn}& & \\
	\cline{2-2}
	& &\ddots& \\
	\cline{4-4}
	& & &\multicolumn{1}{|c|}{dn}\\
	\cline{4-4}
	\end{array}
	\right)}_{D}
\end{align*}
We call this chain $A_1,\dots,A_{\log n}$ the \emph{kernel basis decomposition} of $B$.
Here, $B \prod_i A_i = D$, and each $A_i$ consists of $2^{i-1}$ block matrices on the diagonal of dimension $n / 2^{i-1}$ and degree $d 2^{i-1}$. So while the degree of these blocks doubles, the dimension is halved, which implies that all these $A_i$ can be computed in just $\tilde{O}(dn^\omega)$ 
operations.

Observe that the inverse $B^{-1}$ can be written as $\prod_i A_i D^{-1}$. Also, $D^{-1}$ is a diagonal matrix, and so is easily invertible, i.e. we can write the entries of the inverse in the form of rationals $p/q$ where both $p$ and $q$ are of degree $O(dn)$.
Therefore, we can represent the adjoint via $\adj(B) = B^{-1} \det(B) = \prod_i A_i D^{-1} \det(B)$

To compute $v^\top \adj(B)$ for any degree $d$ vector $v$ in $\tilde{O}(dn^2)$ operations, we must compute $v^\top \prod_i A_i \det(B) D^{-1}$ from left to right. Each vector matrix product with some $A_i$ has degree $d2^{i-1}$ but at the same time the dimension of the diagonal blocks is only $n / 2^{i-1}$, hence each product requires only $\tilde{O}(dn^2)$ field operations. Scaling by $\det(B)$ and dividing by the entries of $D$ also requires only $O(dn^2)$ as their degrees are bounded by $O(dn)$. This gives us \Cref{thm:kernelBasePreprocessing}.

The idea of \Cref{thm:extendedKernelBasePreprocessing} is to explicitly precompute a prefix $P = \prod_{i\le k} A_i$ from the factors of $\prod_i A_i \det(B) D^{-1}$. This increases the preprocessing time but at the same time allows us to compute 
$\adj(B)_{i,j} = \vec{e}_i^\top P \prod_{i>k} A_i \det(B) D^{-1} \vec{e}_j$ faster.

\subsection{Comparison with previous works}
\label{sub:compare}

\paragraph{Previous dynamic matrix algorithms.}

In contrast to algorithms for sensitive oracles that handle a single batch of updates, \emph{dynamic} algorithms must handle an (infinite) sequence of updates. The techniques we used for our sensitive distance/reachability oracles are motivated from techniques developed for dynamic algorithms which we will discuss below.

There is a line of work initiated by Sankowski \cite{Sankowski04,Sankowski07,BrandNS18}
on maintaining inverse or adjoint of a dynamic matrix whose entries
are field elements, not polynomials as in our setting. Let us call
such matrix a \emph{non-polynomial matrix}. By the similar reductions
for obtaining applications on weighted graphs in this paper, dynamic
non-polynomial matrix algorithms imply solutions to many dynamic algorithms
on \emph{unweighted} graphs. 

Despite the similarity of the results and applications, there is
a sharp difference at the core techniques of our algorithm for polynomial
matrices and the previous algorithms for non-polynomial matrices.
The key to all our results is fast preprocessing time. By using
the kernel basis decomposition \cite{ZhouLS15}, we do not need to
explicitly write down the adjoint of a polynomial matrix, which takes
$\Omega(dn^{3})$ operations if the matrix has size $n\times n$ and
degree $d$. This technique is specific for polynomial matrix and
does not have a meaningful counterpart for a non-polynomial matrix. 
On the contrary, algorithms for non-polynomial matrices from \cite{Sankowski04,Sankowski07,BrandNS18}
just preprocess a matrix in a trivial way. That is, they compute the
inverse and/or adjoint explicitly in $O(n^{\omega})$ operations.
Their key contribution is how to handle update in $o(n^{2})$ operations.

We remark that Sankowski \cite{Sankowski05} did obtain a dynamic
polynomial matrix algorithm by extending previous dynamic non-polynomial
matrix algorithms. However, there are two limitations to this approach.
First, the algorithm requires a matrix of the form $(\I-X\cdot A)$
where $A\in\F[X]^{n\times n}$. This restriction excludes some applications
including distances on graphs with zero or negative weights 
because we cannot use the reduction \Cref{lem:shortestDistanceReduction}.
Second, the cost of the algorithm is multiplied by the degree of the
adjoint matrix which is $O(dn)$ if $A$ has degree $d$. Hence, just
to update one entry, this takes $O(dn\times n^{1.407})$ operations\footnote{The current best algorithm \cite{BrandNS18} takes $O(n^{1.407})$ operations
	to update one entry of a non-polynomial matrix.}.
This is already slower than the time for computing from scratch an
entry of adjoint/inverse $\tilde{O}(dn^{\omega})$ using static algorithms\footnote{The first limitation explains why there is only one application in
	\cite{Sankowski05}, which is to maintain distances on unweighted
	graphs. To bypass the second limitation, Sankowski \cite{Sankowski05}
	``forces'' the degrees to be small by executing all arithmetic operations
	under modulo $X^{k}$ for some small $k$. A lot of information about
	the adjoint is lost from doing this. However, for his specific application,
	he can still return the queried distances by combining with other
	graph-theoretic techniques. }.

\paragraph{Previous sensitive distance oracles.}

Previous sensitive distance oracles such as \cite{WeimannY13,GrandoniW12} also use fast matrix-multiplication, but only use it for computing a fast min-plus matrix product in a black box manner. All further techniques used by these algorithms are graph theoretic.

Our shift from graph-theoretic techniques to a purely algebraic algorithm is the key 
that enables us to support large sets of updates. 
Let us explain why previous techniques can inherently handle only small number of deletions.
Their main idea is to sample many smaller subgraphs in the preprocessing. 
To answer a query in the updated graph, their algorithms simply look for a subgraph $H$ where 
(i) all deleted edges were not even in $H$ from the beginning, and
(ii) all edges in the new shortest path are in $H$. 
To argue that $H$ exists with a good probability, the number of deletions cannot be more than $\log n$ where $n$ is the number of nodes.
That is, these algorithms do not really re-compute the new shortest paths, instead they pre-compute subgraphs that ``avoid'' the updates.

Purely algebraic algorithms such as ours (and also \cite{Sankowski05,BrandNS18,BrandN19}) can overcome the limit on deletions naturally. For an intuitive explanation consider the following simplified example for unweighted graphs: Let $A$ be the adjacency matrix of an unweighted graph, then the polynomial matrix $(\I-X \cdot A)$ has the following inverse when considering the field of formal power series: $(\I - X \cdot A)^{-1} = \sum_{k\ge0} X^k A^k$ (this can be seen by multiplying both sides with $\I - X \cdot A$). This means the coefficient of $X^k$ of $(\I - X \cdot A)^{-1}_{i,j}$ is exactly the number of walks of length $k$ from $i$ to $j$. So the entry $(\I - X \cdot A)^{-1}_{i,j}$ does not just tell us the distance between $i$ and $j$, the entry actually encodes \emph{all possible walks} from $i$ to $j$. Thus finding a replacement path, when some edge is removed, becomes very simple because the information of the replacement path is already contained in entry $(\I - X \cdot A)^{-1}_{i,j}$. The only thing we are left to do is to remove all paths from $(\I - X \cdot A)^{-1}_{i,j}$ that use any of the removed edges. This is done via cancellations caused by applying the Sherman-Morrison formula.

Our algorithm exploits the adjoint instead of the inverse, but the interpretation is similar since for invertible matrices the adjoint is just a scaled inverse: $\adj(M) = M^{-1} \det(M)$. We also do not perform the computations over $\Z[X]$, but $\F[X]$ to bound the required bit-length to represent the coefficients.

\subsection{Organization}

We first introduce relevant notations, definitions, and some known reductions in \Cref{sec:preliminaries}.
We construct the adjoint oracles from \Cref{thm:kernelBasePreprocessing,thm:extendedKernelBasePreprocessing} based on kernel basis decomposition in \Cref{sec:inverseOracle}.
Finally, we 
show our algorithms for maintaining adjoint of polynomial matrices in \Cref{sec:dyn_adjoint}, where we will also apply the reductions to get our distance and reachability oracles \Cref{thm:sense st,thm:reach}.

\section{Preliminaries}
\label{sec:preliminaries}

\paragraph{Complexity Measures}

Most of our algorithms work over any field $\F$ and their complexity is measured in the number of arithmetic operations performed over $\F$, i.e. the \emph{arithmetic complexity}. This does not necessarily equal the \emph{time complexity} of the algorithm as one arithmetic operation could require more than $O(1)$ time, e.g. very large rational numbers could require many bits for their representation. This is why our algebraic lemmas and theorems will always state ``in $O( \cdot)$ operations" instead of ``in $O( \cdot )$ time".

For the graph applications however, when having an $n$ node graph, we will typically use the field $\Z_p$ for some prime $p$ of order $n^c$ for some constant $c$. This means each field element requires only $O(\log n)$ bits to be represented and all field operations can be performed in $O(1)$ time in the standard model (or $\tilde{O}(1)$ bit-operations).

\paragraph{Notation: Identity and Submatrices}
The identity matrix is denoted by $\I$.

Let $I, J \subset [n] := \{1,...,n\}$ and $A$ be a $n \times n$ matrix, then the term $A_{I,J}$ denotes the submatrix of $A$ consisting of the rows $I$ and columns $J$. For some $i \in [n]$ we may also just use the index $i$ instead of $\{i\}$. The term $A_{[n],i}$ thus refers to the $i$th column of $A$.

\paragraph{Matrix Multiplication}

We denote with $O(n^\omega)$ the arithmetic complexity of multiplying two $n \times n$ matrices.
Currently the best bound is $\omega < \matrixExponent$ \cite{Gall14a,Williams12}.

\paragraph{Polynomial operations}

Given two polynomials $p, q \in \F[X]$ with $\deg(p), \deg(q) \le d$, we can add and subtract the two polynomials in $O(d)$ operations in $\F$. We can multiply the two polynomials in $O(d \log d)$ using fast-fourier-transformations, likewise dividing two polynomials can be done in $O(d\log d)$ as well \cite[Section 8.3]{AhoHU74}. Since we typically hide polylog factors in the $\tilde{O}( \cdot )$ notation, all operations using degree $d$ polynomials from $\F[X]$ can be performed in $\tilde{O}(d)$ operations in $\F$.

\paragraph{Polynomial Matrices}

We will work with polynomial matrices/vectors, so matrices and vectors whose entries are polynomials.
We define for $M \in \F[X]^{n \times m}$ the degree $\deg(M) := \max_{i,j} \deg(M_{i,j})$. Note that a polynomial matrix $M \in \F[X]^{n \times n}$ with $\det(A) \neq 0$ might not have an inverse in $\F[X]^{n\times n}$ as $\F[X]$ is a ring. However, the inverse $M^{-1}$ does exist in $\F(X)^{n \times n}$ where $\F(X)$ is the field of rational functions.

\paragraph{Adjoint of a Matrix}

The adjoint of an $n \times n$ matrix $M$ is defined as $\adj(M)_{i,j} = (-1)^{i+j} \det(M_{[n]\setminus j, [n] \setminus i})$.
In the case that $M$ has non-zero determinant, we have $\adj(M) = \det(M) \cdot M^{-1}$.
Note that in the case of $M$ being a degree $d$ polynomial matrix, we have $\adj(M) \in \F[X]^{n \times n}$ and $\deg(\adj(M)) < nd$.

\paragraph{Graph properties from polynomial matrices}

Polynomial matrices can be used to obtain graph properties such as the distance between any pair of nodes:

\begin{lemma}[{\cite[Theorem 5 and Theorem 7]{Sankowski05ESA}}]\label{lem:shortestDistanceReduction}
	
	Let $\F := \Z_p$ be a field of size $p \sim n^c$ for some constant $c > 1$ and let $G$ be a graph with $n$ nodes and integer edge weights $(c_{i,j})_{1\le i,j \le n} \in [-W, W]$.
	
	Let $A \in \F[X]^{n \times n}$ be a polynomial matrix, where $A_{i,i} = X^W$ and $A_{i,j} = a_{i,j} X^{W+c_{i,j}}$ and each $a_{i,j} \in \F$ is chosen independently and uniformly at random.
	
	\begin{itemize}
		\item If $G$ contains no negative cycle, then the smallest degree of the non-zero monomials of $\adj(A)_{i,j}$ minus $W(n-1)$ is the length of the shortest path from $i$ to $j$ in $G$ with probability at least $1-n^{1-c}$.
		
		\item Additionally with probability at least $1-n^{1-c}$, the graph $G$ has a negative cycle, if and only if $\det(A)$ has a monomial of degree less than $Wn$.
	\end{itemize}

\end{lemma}

\section{Adjoint Oracle}
\label{sec:inverseOracle}

In this section we will outline how the adjoint oracle \Cref{thm:kernelBasePreprocessing} by \cite{ZhouLS15} can be extended to our \Cref{thm:extendedKernelBasePreprocessing}.

Unfortunately this new result is not a blackbox reduction, instead we have to fully understand and exploit the properties of the algorithm presented in \cite{ZhouLS15}.
This is why a formally correct proof of \Cref{thm:extendedKernelBasePreprocessing} requires us to repeat many definitions and lemmas from \cite{ZhouLS15}. Such a formally correct proof can be found in subsection \ref{app:adjointOracle}. We will start with a high level description based on the high level idea of \Cref{thm:kernelBasePreprocessing} presented in \Cref{sec:overview}.

\subsection{Extending the Oracle to Element Queries}
\label{sub:highLevelAdjoint}

We will now outline how the data-structure of kernel-bases, presented in \Cref{sec:overview}, can be used for faster element queries to $\adj(B)$.
Remember that \Cref{thm:kernelBasePreprocessing} was based on representing $\adj(B) = \prod_{i=1}^{ \log n } A_i D^{-1} \det(B)$, where each $A_i$ consists of $2^{i-1}$ diagonal blocks of size $n/2^{i-1} \times n / 2^{i-1}$ and is of degree $d2^{i-1}$.

The idea for \Cref{thm:extendedKernelBasePreprocessing} is very simple: Choose some $0 \le \mu \le 1$ and $k$ such that $2^k = n^\mu$, then during the pre-processing compute the kernel-base decomposition $\adj(B) = \prod_i A_i D^{-1} \det(B)$ and pre-compute the product $M := \prod_{i=1}^k A_i$ explicitly.
When an entry $(i,j)$ of $\adj(B)$ is required, we only have to compute $\vec{e}_i^\top M \prod_{i>k} A_i \vec{e}_j D_{j,j}^{-1} \det(B)$.

\paragraph{Complexity of the Algorithm}

Let $A_i^{[j]}$ be the $(n/2^{i-1}) \times (n/2^{i-1})$ matrix obtained when setting all of $A_i$ to 0, except for the diagonal block that includes the $j$th column. We will now argue, that $\vec{e}_i^\top M \prod_{r>k} A_r \vec{e}_j = \vec{e}_i^\top M \prod_{r>k} A_r^{[j]} \vec{e}_j$. This equality can be seen by computing the product from right to left:

\begin{itemize}
\item Consider the right-most product $A_{\log n} \vec{e}_j$. The vector $\vec{e}_j$ is non-zero only in the $j$th row, so only the $j$th column of $A_{\log n}$ matters, hence $A_{\log n} \vec{e}_j = A_{\log n}^{[j]} \vec{e}_j$. 

\item Consider the product $A_i A_{i+1}^{[j]}$. The matrix $A_{i+1}^{[j]}$ has few non-zero rows, so most columns of $A_i$ will be multiplied by zero and we thus most entries of $A_i$ do not matter for computing the product. Note that all entries of $A_i$ that \emph{do} matter (i.e. are multiplied with non-zero entries of $A_{i+1}^{[k]}$) are inside the block $A_i^{[j]}$, because of the recursive structure of the matrices (i.e. the blocks of $A_{i+1}$ are obtained by splitting the blocks of $A_i$), see for instance \Cref{fig:oracle}. This leads to $A_i A_{i+1}^{[j]}=A_i^{[j]} A_{i+1}^{[j]}$.
\end{itemize}
By induction we now have $\vec{e}_i^\top M \prod_{r>k} A_r \vec{e}_j = \vec{e}_i^\top M \prod_{r>k} A_r^{[j]} \vec{e}_j$.

\begin{figure}
\center
\includegraphics[trim={150 150 170 100},clip,scale=0.5]{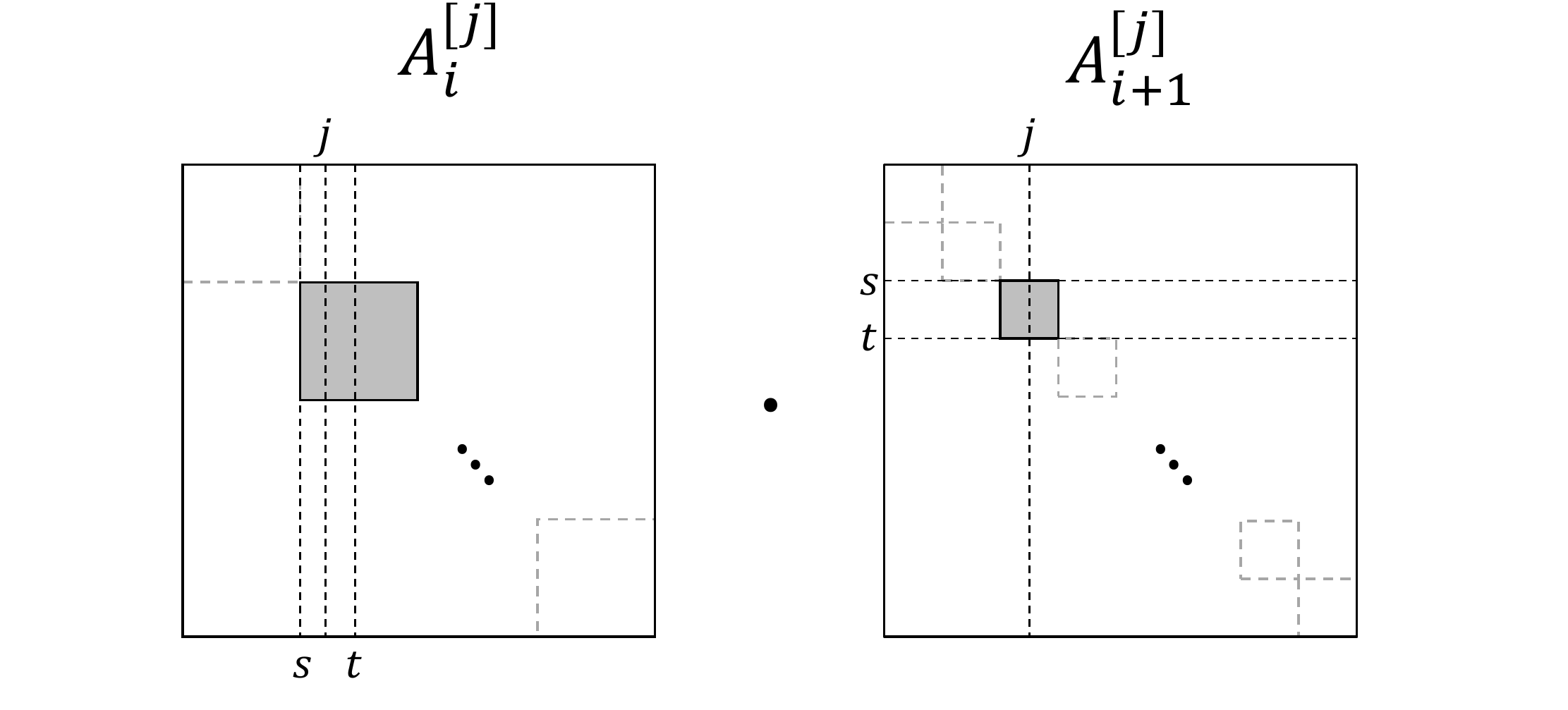}
\caption{\label{fig:oracle}
Dark grey boxes represent non-zero entries. The grey dotted squares represent non-zero entries of $A_i$ that are set to zero in $A_i^{[j]}$. Here we see that only rows with index $s,...,t$ of $A_{i+1}^{[j]}$ are non-zero, so when computing $A_i A_{i+1}^{[j]}$ only the columns of $A_i$ with index $s,...,t$ need to be considered. These columns are identical in $A_i$ and $A_i^{[j]}$ so $A_i A_{i+1}^{[j]} = A_i^{[j]}A_{i+1}^{[j]}$.}
\end{figure}

The complexity of computing this product is very low, when multiplied from left to right.
Consider the first products $\vec{e}_i M A_{k+1}^{[j]}$. The degree of matrix $M$ and $A_{k+1}^{[j]}$ are both bounded by $O(dn2^k)$. The matrix $A_{k+1}^{[j]}$ is 0 except for a $(n / 2^k) \times (n / 2^k)$ block on the diagonal. Hence, this first product requires $\tilde{O}(d2^k (n/2^k)^2)$ field operations.

All products after this require fewer operations: On one hand the degree of vector and matrix double after each product, on the other hand the dimension of the non-zero block of $A_r^{[j]}$ is halved. Since the complexity of the vector matrix product scales linearly in the degree but quadratic in the dimension, the complexity is bounded by the initial product $\vec{e}_i M A_{k+1}^{[j]}$.
The query complexity is thus $\tilde{O}(d2^k (n / 2^k)^2) = \tilde{O}(dn^{2-\mu})$.

This is only a rough simplification of how the algorithm works. For instance the degrees of the $A_1,...A_{\log n}$ are not simple powers of 2, instead only the average degree is bounded by a power of 2. Likewise the dimension $n$ and the size of the diagonal blocks do not have to be a power of two.

\subsection{Formal Proof of the Adjoint Element Oracle}
\label{app:adjointOracle}

Before we can properly prove our \Cref{thm:extendedKernelBasePreprocessing}, we first have to define/cite some terminology and lemmas from \cite{ZhouLS15}, as our \namecref{thm:extendedKernelBasePreprocessing} is heavily based on their result.

First we will define the notation of shifted column degrees. Shifted column degrees can be used to formalize how the degree of a vector changes when multiplying it with a polynomial matrix.

\begin{definition}[{\cite[Section 2.2]{ZhouLS15}}]
Let $\F$ be some field, $M \in \F[X]^{n \times m}$ be some polynomial matrix and let $\vec{s} \in \N^n$ be some vector.

Then the $\vec{s}$-shifted column degrees of $M$ is defined via:
$$
\cdeg_{\vec{s}}(M)_j := \max_{i=1,...,n} \vec{s}_i + \deg(M_{i,j}) \text{ for }j=1,...,m
$$
\end{definition}

\begin{lemma}\label{lem:vectorMatrixComplexity}
Let $\F$ be some field, $M \in \F[X]^{n \times m}$ be some polynomial matrix and let $\vec{s} \in \N^n$ be some vector. Further let $\vec{v} \in \F[X]^n$ be a polynomial vector where $\deg(\vec{v}_i) \le \vec{s}_i$ for $i=1,...,n$.

Then $\deg((\vec{v}^\top M)_j) \le \cdeg_{\vec{s}}(M)_j$ and $\vec{v}^\top M$ can be computed in $\tilde{O}(n \sum_{j=1}^m \cdeg_{\vec{s}}(M)_j)$.
\end{lemma}

\begin{proof}
Multiplying $\vec{v}_i M_{i,j}$ requires $\tilde{O}(\deg(\vec{v}_i)+\deg(M_{i,j}))$ field operations. Hence the total cost becomes $$
\sum_{j=1}^m \sum_{i=1}^n \tilde{O}\left(\deg(\vec{v}_i)+\deg(M_{i,j})\right)
\le \tilde{O}\left(\sum_{j=1}^m \sum_{i=1}^n \vec{s}_i+\deg(M_{i,j})\right)
\le \tilde{O}\left(n \sum_{j=1}^m \cdeg_{\vec{s}}(M)_j \right).
$$
\end{proof}

We will now give a formal description of the data-structure constructed in \cite{ZhouLS15}. The following definitions and properties hold throughout this entire section.

Let $B \in \F[X]^{n \times n}, \vec{s} \in \N^n$ such that $\cdeg_{\vec{0}}(B)_j \le \vec{s}_j$, so $\vec{s}_j$ bounds the maximum degree in the $j$th column of $M$ (also called the column degree of $M$). Let $d := \sum_i \vec{s}_i / n$ be the average column degree of $B$.

In \cite{ZhouLS15} they construct in $\tilde{O}(dn^\omega)$ field operations a chain of matrices $A_1,...,A_{\lceil \log n \rceil} \in \F[X]^{n \times n}$ and a diagonal matrix $D \in \F[X]^{n \times n}$ such that $\adj(B) = (\prod_i A_i) \det(B) D^{-1}$.

Here the matrices $(A_{i+1})_{i=0...\lceil \log n \rceil-1}$ are block matrices consisting each of $2^i$ diagonal blocks, i.e. 
$$A_{i+1} = \diag(A_{i+1}^{(1)},...,A_{i+1}^{(2^i)}).$$
The number of rows/column of each $A_{i+1}^{(k)}$ is $n/2^i$ upto a factor of 2. (Note that $A_{i}^{(k)}$ refers to the $k$th block on the diagonal of $A_i$, not to be confused with our earlier definition of $A_i^{[j]}$ in subsection \ref{sub:highLevelAdjoint}).

Remember from the overview (\Cref{sec:overview}) that each $A_{i}^{(k)}$ consists of two kernel bases, so each of these diagonal block matrices consists in turn of two matrices (kernel bases)
$$A_{i+1}^{(k)} = [N_{i+1,l}^{(k)},N_{i+1,r}^{(k)}].$$
Here $l$ and $r$ are not variables but denote the \emph{left} and \emph{right} submatrix.

We also write $M_i$ for the partial product $M_i := \prod_{k=1}^i A_k$. Each $M_i$ can be decomposed into 
$$M_i = [M_i^{(1)},...,M_i^{(2^i)}]$$
where each $M_i^{(k)}$ has $n$ rows and the number of columns in $M_i^{(k)}$ corresponds to the number of columns in $A_{i+1}^{(k)}$. We can compute $M_i$ as follows:
\begin{align}
M_{i+1}^{(2k-1)} = M_i^{(k)} N_{i+1,l}^{(k)} \text{ and } M_{i+1}^{(2k)} = M_i^{(k)}N_{i+1,r}^{(k)} \label{eqn:computeM}
\end{align}

We have the following properties for the degrees of these matrices:

\begin{lemma}[Lemma 10 in \cite{ZhouLS15}]\label{lem:degreeBound}
Let $m \times m$ be the dimension of $M_i^{(k)}$, let $m_l$ and $m_r$ be the number of columns in $N_{i+1,l}^{(k)}$ and $N_{i+1,r}^{(k)}$ respectively and let $\vec{t} := \cdeg_{\vec{s}}(M_i^{(k)})$, then
\begin{itemize}
\item $\sum_{j=1}^m \vec{t}_j \le \sum_{j=1}^n \vec{s}_j = dn$, and
\item $\sum_{j=1}^{m_l} \cdeg_{\vec{t}}(N_{i+1,l}^{(k)})_j \le \sum_{j=1}^n \vec{s}_j = dn$ and $\sum_{j=1}^{m_r} \cdeg_{\vec{t}}(N_{i+1,r}^{(k)})_j \le \sum_{j=1}^n \vec{s}_j = dn$\\
(which also implies $\sum_{j=1}^m \cdeg_{\vec{t}}(A_{i+1}^{(k)})_j \le 2dn$)
\end{itemize}
\end{lemma}

\begin{lemma}[Lemma 11 in \cite{ZhouLS15}]\label{lem:computeMcomplexity}

For a given $i$ and $k$ the matrix multiplications in \eqref{eqn:computeM} can be done in $\tilde{O}(n(n/2^i)^{\omega-1}(1+d 2^i)))$ field operations.

\end{lemma}

We now have defined all the required lemmas and notation from \cite{ZhouLS15} and we can now start proving \Cref{thm:extendedKernelBasePreprocessing}. 

The following lemma is analogous to \Cref{lem:computeMcomplexity}, though now we want to compute only one row of product \eqref{eqn:computeM}. This lemma will bound the complexity of the query operation in \Cref{thm:extendedKernelBasePreprocessing}.

\begin{lemma}\label{lem:computePartialRow}
Let $v^\top := \vec{e}_r^\top M_i^{(k)}$ be some row of $M_i^{(k)}$, then we can compute $v^\top N_{i+1,c}^{(k)}$ for both $c =l,r$ in $\tilde{O}(dn^2/2^i)$.
\end{lemma}

\begin{proof}
The matrices $N_{i+1,l}^{(k)}$ and $N_{i+1,r}^{(k)}$ form the matrix $A_{i+1}^{(k)}$, so we instead just compute the product $\vec{v}^\top A_{i+1}^{(k)}$. The matrix $A_{i+1}^{(k)}$ is of size $(n/2^i) \times (n/2^i)$ (up to a factor of 2).
Let $\vec{t} := \cdeg(\vec{v}^\top)$ then by \Cref{lem:degreeBound} we know $\sum_j \cdeg_{\vec{t}}(A_{i+1}^{(k)})_j \le 2 \sum_j \vec{s}_j = 2dn$ and $\sum_j \vec{t}_j \le dn$.

By \Cref{lem:vectorMatrixComplexity} the cost of computing $\vec{v}^\top A_{i+1}^{(k)}$ is $
\tilde{O}(n/2^i \sum_{j=1}^m \cdeg_{\vec{t}}(A_{i+1}^{(k)})_j )
$, which given the degree bounds can be simplified to $\tilde{O}(dn^2/2^i)$.
\end{proof}

The following lemma will bound the complexity for the pre-processing of \Cref{thm:extendedKernelBasePreprocessing}.

\begin{lemma}\label{lem:adjointPreprocessing}
If we already know the matrix $A_i$ for $i=1,...\lceil \log n \rceil$, then for any $0 \le \mu \le 1$ we can compute $M_{\lceil \log n^\mu \rceil}$ in $\tilde{O}(dn^{\omega(1-\mu)+3\mu})$ field operations.
\end{lemma}

\begin{proof}
To compute $M_{i+1}$, requires to compute all $2^i$ many $M_{i+1}^{(k)}$ for $k=1...2^i$. Assume we already computed matrix $M_i$, then we can compute $M_{i+1}^{(k)}$ for $k=1...2^i$ via \Cref{lem:computeMcomplexity}. Inductively the total cost we obtain is:
$$
\sum_{i=1}^{\lceil \log n^\mu \rceil}\tilde{O}(n(n/2^i)^{\omega-1}(1+d 2^i)) \cdot 2^i = \tilde{O}(dn^{\omega(1-\mu)+3\mu})
$$
\end{proof}

The last lemma we require for the proof of \Cref{thm:extendedKernelBasePreprocessing} is that we can compute the determinant of $B$ in $\tilde{O}(dn^\omega)$ field operations.
\begin{lemma}[\cite{Storjohann03,LabahnNZ17}]\label{lem:polynomialDeterminant}
Let $B \in \F[X]^{n \times n}$ be a matrix of degree at most $d$, then we can compute $\det(B)$ in $\tilde{O}(dn^\omega)$ field operations.
\end{lemma}

\begin{proof}[Proof of \Cref{thm:extendedKernelBasePreprocessing}]
The claim is that for any $0 \le \mu \le 1$ we can, after $\tilde{O}(dn^{\omega(1-\mu)+3\mu})$ pre-processing of $B$, compute any entry $\adj(B)_{i,j}$ in $\tilde{O}(n^{2-\mu})$ operations.

\paragraph{Pre-processing}
We first compute the determinant $\det(B)$ via \Cref{lem:polynomialDeterminant} and construct the chain of matrices $A_1,...A_{\lceil \log n \rceil}$ as in \cite{ZhouLS15} in $\tilde{O}(dn^\omega)$, then we compute $M_{\lceil \log n^\mu \rceil}$ in $\tilde{O}(dn^{\omega(1-\mu)+3\mu})$ using \Cref{lem:adjointPreprocessing}.

\paragraph{Queries}
When answering a query for $\adj(A)_{i,j}$ we compute one entry $(M_{\lceil \log n \rceil})_{i,j}$, multiply it with $\det(A)$ and divide it by $D_{j,j}$, because $\adj(A) = (\prod_{i=1}^{\lceil \log n \rceil} A_i) \det(A) D^{-1} = M_{\lceil \log n \rceil} \det(A) D^{-1}$.

Here the expensive part is to compute the entry of $M_{\lceil \log n \rceil}$, which is done by computing one row of $M_{\lceil \log n \rceil}^{(k)}$ for some appropriate $k$. Via \eqref{eqn:computeM}, we know
$$
M_{\lceil \log n \rceil}^{(k)} = M_{\lceil \log n^\mu \rceil} \prod_{t=1}^{\lceil \log n \rceil - \lceil \log n^\mu \rceil}N_{\lceil \log n^\mu \rceil+t,c_t}^{(k_t)}$$
for some sequence $k_t \in [2^{\lceil \log n^\mu \rceil+t}]$, $c_t \in \{l,r\}$.

So we only have to compute the product of the $i$th row of $M_{\lceil \log n^\mu \rceil}$ with a sequence of $(N_{\lceil \log n^\mu \rceil+t,c_t}^{(k_t)})_{t=1...\lceil \log n \rceil-1}$ matrices. Computing this product from left to right means we compute the following intermediate results
$$
\vec{e}_i^\top M_{\lceil \log n^\mu \rceil} \prod_{t=1}^{r}N_{\lceil \log n^\mu \rceil+t,c_t}^{(k_t)}
= \vec{e}_i^\top M_{\lceil \log n^\mu \rceil + r}^{(k_r)}
$$
for every $r=1,...,\lceil \log n \rceil - \lceil \log n^\mu \rceil$. So each intermediate result is just the $i$th row of some matrix $M_{\lceil \log n^\mu \rceil + r}^{(k_r)}$. This means such a vector-matrix product can be computed via \Cref{lem:computePartialRow} in
$$
\tilde{O}(dn^2/2^{\lceil \log n^\mu \rceil+r})
$$
Here the first product for $r=1$ is the most expensive and the total cost for all $\lceil \log n \rceil - \lceil \log n^\mu \rceil$ many vector-matrix products becomes $\tilde{O}(dn^{2-\mu})$.

\end{proof}

\section{Sensitive Distance and Reachability Oracles}

In this section we will use the adjoint oracle from \Cref{sec:inverseOracle} to obtain the results presented in \Cref{sec:introduction}. A high-level description of that algorithm was already outlined in \Cref{sec:overview}.

This section is split into two parts: First we will describe in \Cref{sec:dyn_adjoint} our results for maintaining the adjoint of a polynomial matrix, which will conclude with the proof of \Cref{thm:sense st}. The second subsection \ref{app:sensitiveReachabilityOracle} will explain how to obtain the sensitive reachability oracle \Cref{thm:reach}. These graph theoretic results will be stated more accurately than in the overview by adding trade-off parameters and memory requirements.

All proofs in this section assume that the matrix $A$ stays non-singular throughout all updates, which is the case w.h.p for matrices constructed via the reduction of \Cref{lem:shortestDistanceReduction}.

\subsection{Adjoint oracles}
\label{sec:dyn_adjoint}

\begin{theorem}\label{thm:emergencyInverseElement}
Let $\F$ be some field and $A \in \F[X]^{n \times n}$ be a polynomial matrix with $\det(A) \neq 0$ of degree $d$. 
Then for any $0\le \mu \le 1$ we can create a data-structure 
storing $O(dn^{2+\mu})$ field elements 
in $\tilde{O}(dn^{\omega(1-\mu)+3\mu})$ field operations, such that:

We can change $f$ columns of $A$ 
and update our data-structure in 
$\tilde{O}(dn^{2-\mu}f^2+dnf^\omega)$ field operations 
storing additional $O(dnf^2)$ field elements. 
This updated data-structure allows for querying entries of $\adj(A)$ 
in $\tilde{O}(dnf^2+dn^{2-\mu}f)$ field operations 
and queries to the determinant of the new $A$ 
in $O(dn)$ operations.
\end{theorem}

\begin{proof}
We start with the high-level idea: We will express the change of $f$ entries of $A$ via the rank $f$ update $A+UV^\top$, where $U,V \in \F^{n \times f}$ and both matrices have only one non-zero entry per column.
Via \Cref{lem:adjointUpdate} we have for $M := \I \cdot \det(A) + V^\top \adj(A)U$ that
$$
\adj(A+UV^\top ) = \frac{\adj(A)\det(M) - \left(\adj(A) U\right) \: \adj(M) \: \left(V^\top \adj(A)\right)} {\det(A)^f }
$$

Our algorithm is as follows:
\paragraph{Initialization}

During the initialization use \Cref{thm:extendedKernelBasePreprocessing} on matrix $A$ in $\tilde{O}(dn^{\omega+(3-\omega)\mu})$ operations, which will allow us later to query entries $\adj(A)_{i,j}$ in $\tilde{O}(dn^{2-\mu})$ operations. We also compute $\det(A)$ via \Cref{lem:polynomialDeterminant} in $\tilde{O}(dn^\omega)$ operations.

\paragraph{Updating the data-structure}

The first task is to compute the matrix $M$. Note that for element updates to $A$, the matrices $U$ and $V$ have only one non-zero entry per column, so $V^\top \adj(A) U$ is just an $f \times f$ submatrix of $A$ where each entry is multiplied by one non-zero entry of $U$ and $V$. Thus we can compute $M = \I \cdot \det(A) + V^\top \adj(A)U$ in $\tilde{O}(df^2n^{2-\mu})$ operations, thanks to the pre-processing of $A$.
Next, we pre-processing the matrix $M$ using \Cref{thm:kernelBasePreprocessing}, which requires $\tilde{O}(dnf^\omega)$ as the degree of $M$ is bounded by $O(dn)$. We also compute $\det(M)$ in $\tilde{O}(dnf^\omega)$ operations.

Thus the update complexity is bounded by $\tilde{O}(df^2n^{2-\mu} + dnf^\omega)$ operations.

\paragraph{Querying entries}

To query an entry $(i,j)$ of $\adj(A+UV^\top)$ we have to compute
$$
\frac{
\adj(A)_{i,j} \det(M) 
- \left(\vec{e}_i \adj(A) U\right) 
\: \adj(M) \: 
\left(V^\top \adj(A) \vec{e}_j\right)
}{\det(A)^f}.$$

Here $\adj(A)_{i,j} \det(M)$ can be computed in $\tilde{O}(dn^{2-\mu}+dnf)$ operations.
The vectors $\vec{e}_i \adj(A) U$ and $V^\top \adj(A) \vec{e}_j$ are just $f$ entries of
$\adj(A)$ where each entry is multiplied by one non-zero entry of $U$ or $V$, so we can compute them in
$\tilde{O}(dn^{2-\mu}f)$ operations. Multiplying one of these vectors with $\adj(M)$ requires $\tilde{O}(dnf^2)$ operations because of the
pre-processing of $M$ via \Cref{thm:kernelBasePreprocessing}.
The product of $\left(\vec{e}_i \adj(A) U\right) \adj(M)$ with $\left(V^\top \adj(A) \vec{e}_j\right)$ can be computed in
$\tilde{O}(dnf^2)$ operations. Subtracting $\adj(A)_{i,j} \det(M)$ from the product and dividing by $\det(A)^f$ can be done in
$\tilde{O}(dnf)$ operations.

The query complexity is thus $\tilde{O}(dnf^2+dn^{2-\mu}f)$ operations.

\paragraph{Maintaining the determinant}

We have $\det(A+UV^\top) = \det(M) / \det(A)^{f-1}$ 
(as can be seen in the proof of \Cref{sec:proof adjointUpdate}). 
The division requires only $\tilde{O}(dnf)$ operations.

\end{proof}

In this section we will state the result from \Cref{sec:introduction} in a more
formal way. We will state trade-off parameters, memory consumption and we will
separate the sensitive setting into two phases: update and query.

\begin{theorem}[{Corresponds to \Cref{thm:sense st}}]\label{thm:emergencyGraphPairOnly}
Let $G$ be a graph with $n$ nodes and edge weights in $[-W,W]$.
For any $0\le \mu \le 1$ we can create 
in $\tilde{O}(Wn^{\omega(1-\mu)+3\mu})$ time 
a Monte Carlo data-structure 
that requires $O(Wn^{2+\mu} \log n)$ bits of memory, 
such that:

We can then change $f$ edges (additions and removals) and update our
data-structure in $\tilde{O}(Wn^{2-\mu}f^2+Wnf^\omega)$ time using additional
$O(Wnf^2 \log n)$ bits of memory. This updated data-structure allows for querying the
distance between any pair of nodes in $\tilde{O}(Wnf^2+Wn^{2-\mu}f)$ time,
or report that there exists a negative cycle in $O(1)$ time.
\end{theorem}

\begin{proof}[Proof of \Cref{thm:emergencyGraphPairOnly}]
	
We construct the matrix as specified in \Cref{lem:shortestDistanceReduction}.
Since the field size is polynomial in $n$, the arithmetic operations can be
executed in $O(1)$ time in the standard model and saving one field element
requires $O(\log n)$ bits.
	
The determinant of the matrix is computed during the update, where we will
also check for negative cycles by checking if the determinant has a non-zero
monomial of degree less than $Wn$. This way we can answer the existence of a
negative cycle in $O(1)$ time, when a query to a distance is performed.
	
\Cref{thm:emergencyGraphPairOnly} is now implied by \Cref{thm:emergencyInverseElement}.
	
\end{proof}

\subsection{Sensitive reachability oracle}
\label{app:sensitiveReachabilityOracle}

So far we have only proven the distance applications. Using the same techniques
the result can also be extended to a sensitive reachability oracle.

Reachability can be formulated as a distance problem where every edge has cost 0.
In that case the matrix constructed in \Cref{lem:shortestDistanceReduction} has
degree 0, so the matrix is in $\F^{n \times n}$ and we do not have to bother with
polynomials anymore. The algorithm for the following result is essentially the
same as \Cref{thm:emergencyInverseElement}, but since we no longer have polynomial
matrices, we no longer require more sophisticated adjoint oracles and can instead
use simpler subroutines.

\begin{theorem}[{Corresponds to \Cref{thm:reach}}]\label{thm:reachability}
Let $G$ be a graph with $n$ nodes.
We can construct in $O(n^{\omega})$ time a Monte Carlo data-structure 
that requires $O(n^{2} \log n)$ bits of memory, such that:

We can change any $f$ edges (additions and removals) and update our data-structure
in $O(f^\omega)$ time using additional $O(f^2 \log n)$ bits of memory. This updated
data-structure allows for querying the reachability between any pair of nodes in
$O(f^2)$ time.
\end{theorem}

\begin{proof}

We will now construct a data-structure that, after some initial pre-processing of
some matrix $A \in \F^{n \times n}$, allows us to quickly query elements of
$\adj(A+UV^\top )$ for any sparse $U,V \in \F^{n \times f}$ where $U$ and $V$ have
at most one non-zero entry per column.

\paragraph{Pre-processing}
We compute $\det(A)$ and $\adj(A)$ explicitly in $O(n^\omega)$ field operations.

\paragraph{Update}
We receive matrices $U,V \in \F^{f \times n}$, where each column has only one
non-zero entry. We compute $V^\top \adj(A) U$ in $O(f^2)$ operations because of the
sparsity of $U$ and $V$.
This means we can now compute for $M := \I \cdot \det(A) + V^\top \adj(A) U$ the
matrix $\adj(M)$ and $\det(M)$ in $O(f^\omega)$ operations.

\paragraph{Query}
When querying an entry $(i,j)$ of $\adj(A+U V^\top)$ we have to compute:
\begin{align*}
\frac{
\adj(A)_{i,j} \det(M) - \adj(A)_{i,[n]}U \: \adj(M) \: V^\top \adj(A)_{[n],j}
}{
\det(A)^f
}
\end{align*}

The vectors $\vec{u}^\top:=\adj(A)_{i,[n]}U$ and
$\vec{v} := V^\top \adj(A)_{[n],j}$ can be computed in $O(f)$ operations 
as they are just $O(f)$ entries of the adjoint, 
each multiplied by one non-zero entry of $U$ or $V$. 
The product $\vec{u}^\top \adj(M) \vec{v}$ can be computed in $O(f^2)$ operations.

Thus the entry of $\adj(A+UV^\top)$ can be computed in $O(f^2)$ operations.

\end{proof}

\section{Open Problems}
\label{sec:open_problems}

We present the first graph application 
of the kernel basis decomposition by \cite{JeannerodV05,ZhouLS15},
and we are interested in seeing if there are further uses of that technique 
outside of the symbolic computation area.

Our sensitive distance oracle has subqubic preprocessing time and subquadratic query time.
When restricting to only a single edge deletion, Grandoni and V. Williams
where able to improve this to subqubic preprocessing and sublinear query time \cite{GrandoniW12}.
An interesting open question is whether a similar result can be obtained for multiple edge deletions.
Even supporting just two edge deletions 
with subqubic preprocessing and sublinear query time 
would a good first step.
Alternatively, disproving the existence of such a data-structure would also be interesting.
Currently the best (conditional) lower bound 
by V. Williams and Williams \cite{VWilliamsW18,HenzingerL0W-ITCS17} refutes such an algorithm,
if its complexity has a $\text{polylog}(W)$ dependency on the largest edge weight $W$.
For algorithms with $\text{poly}(W)$ dependency no such lower bound is known.

All our data-structures maintain only the distance, 
or in case of reachability, return a boolean answer. 
So another open problem would be to find a data-structure,
that does not just return the distance, but also the actual path.

Finally, our data-structures are randomized Monte-Carlo. 
We wonder if there is a deterministic equivalent.
In case of the previous result by Weimann and Yuster \cite{WeimannY13},
a deterministic equivalent was found by Alon et al.~\cite{AlonCC19}.

\section*{Acknowledgment}

This project has received funding from the European Research Council (ERC) under the European 
Unions Horizon 2020 research and innovation programme under grant agreement No 715672.
This work was mostly done while Thatchaphol Saranurak was at KTH Royal Institute of Technology.

\bibliographystyle{alpha}
\bibliography{bibliography}

\newpage

\appendix

\section{Proof of \Cref{lem:adjointUpdate}}
\label{sec:proof adjointUpdate}

\begin{proof}[Proof of \Cref{lem:adjointUpdate}]
	Via the Sherman-Morrison-Woodbury formula we have
	\begin{align*}(A + UV^\top )^{-1} = A^{-1} - A^{-1} U (\I + V^\top A^{-1} U)^{-1} V^\top A^{-1}
	\end{align*}
	and via the Sylvester determinant identity $\det(\I+AB) = \det(\I + BA)$ we have
	\begin{align*}
	\det(A + UV^\top) =  \det(A) \cdot \det(\I + A^{-1} UV^\top) = \det(A) \cdot \det(\I + V^\top A^{-1} U).
	\end{align*}
	This allows us to write the determinant of $M := \I \cdot \det(A) + V^\top \adj(A) U$ as follows:
	\begin{align*}
	\det(M) &= \det(\I \cdot \det(A) + V^\top \adj(A) U)\\
	&=
	\det(A)^f \cdot \det(\I + V^\top A^{-1} U)\\
	&= \det(A+UV^\top ) \cdot \det(A)^{f-1}
	\end{align*}
	Which yields $\adj(A) \cdot \det(M)\det(A)^{-f} = A^{-1}  \det(A+UV^\top )$ because $\adj(A) = A^{-1} \det(A)$.
	Likewise we obtain:
	\begin{align*}
	&\frac{\left(\adj(A) U\right) \: \adj(M) \: \left(V^\top \adj(A)\right)} {\det(A)^f }\\
	=&
	\frac{\det(M)\det(A)^2\left(A^{-1} U\right) \: (\I \cdot \det(A) + V^\top \adj(A) U)^{-1} \: \left(V^\top A^{-1}\right)} {\det(A)^{f} }\\
	=&
	\frac{\det(M)\det(A)\left(A^{-1} U\right) \: (\I + V^\top A^{-1} U)^{-1} \: \left(V^\top A^{-1}\right)} {\det(A)^{f} }\\
	=&
	\det(A+UV^\top)\left(A^{-1} U\right) \: (\I + V^\top A^{-1} U)^{-1} \: \left(V^\top A^{-1}\right)
	\end{align*}
	Thus we obtain \Cref{lem:adjointUpdate} by multiplying the Sherman-Morrison-Woodbury identity with $\det(A+UV^\top)$.
\end{proof}

\end{document}